\documentclass[11pt]{article}

\usepackage[final]{acl}

\usepackage{times}
\usepackage{latexsym}

\usepackage[T1]{fontenc}

\usepackage[utf8]{inputenc}

\usepackage{microtype}

\usepackage{inconsolata}

\usepackage{graphicx}

\usepackage{amsmath}
\usepackage{amssymb}
\usepackage{mathtools}
\usepackage{amsthm}
\usepackage{bm}
\usepackage{multirow}
\usepackage{hyperref}
\usepackage{url}
\usepackage{booktabs}
\usepackage{amsfonts}
\usepackage{nicefrac}
\usepackage{microtype}
\usepackage{xcolor}
\usepackage{bbding}
\usepackage{wrapfig}
\usepackage{subcaption}
\usepackage{bbm}

\theoremstyle{plain}
\newtheorem{theorem}{Theorem}[section]

\theoremstyle{remark}

\newcommand{\x}{\bm{x}}

\renewcommand{\k}{\bm{k}}

\renewcommand{\k}{\bm{k}}

\title{More Haste, Less Speed: Weaker Single-Layer Watermark Improves Distortion-Free Watermark Ensembles}

\author{Ruibo Chen$^1$, Yihan Wu$^1$, Xuehao Cui$^1$, Jingqi Zhang$^2$, Heng Huang$^1$ \\
  $^1$University of Maryland, College Park, $^2$National University of Singapore\\
  \texttt{\{rbchen,ywu42,cedrus,heng\}@umd.edu},   \texttt{zhangjingqi@u.nus.edu} \\
}

\begin{document}
\maketitle
\begin{abstract}
Watermarking has emerged as a crucial technique for detecting and attributing content generated by large language models. While recent advancements have utilized watermark ensembles to enhance robustness, prevailing methods typically prioritize maximizing the strength of the watermark at every individual layer. In this work, we identify a critical limitation in this "stronger-is-better" approach: strong watermarks significantly reduce the entropy of the token distribution, which paradoxically weakens the effectiveness of watermarking in subsequent layers. We theoretically and empirically show that detectability is bounded by entropy and that watermark ensembles induce a monotonic decrease in both entropy and the expected green-list ratio across layers. To address this inherent trade-off, we propose a general framework that utilizes weaker single-layer watermarks to preserve the entropy required for effective multi-layer ensembling. Empirical evaluations demonstrate that this counter-intuitive strategy mitigates signal decay and consistently outperforms strong baselines in both detectability and robustness.
\end{abstract}

\section{Introduction}
Large language models have demonstrated unprecedented capabilities in generating high-quality text, raising significant concerns regarding copyright infringement, misinformation, and the potential misuse of AI-generated content~\citep{goldstein2023generative,weidinger2021ethical,carlini2022quantifying,DBLP:conf/icml/ChenWGH25}. To address these challenges, watermarking~\citep{Aaronson2022,kirchenbauer2023watermark,zhao2023provable,liu2024adaptive,chen2025watermark,cui2026mc} has emerged as a crucial technique for detecting and attributing machine-generated text. Among various approaches, \textit{distortion-free watermarking}~\citep{Aaronson2022, christ2023undetectable, hu2023unbiased,wu2026ensemble} has garnered particular attention because it mathematically guarantees that the watermarked text distribution matches the original model distribution in expectation, thereby better preserving generation quality.

Recent advancements have shifted focus from single-layer schemes to \textit{watermark ensembles}~\citep{dathathri2024scalable,feng2025bimark,wu2026ensemble}. By applying watermarking functions recursively across multiple layers, ensemble methods aim to aggregate statistical signals to enhance detectability and robustness. The prevailing design philosophy in these methods typically prioritizes maximizing the strength of the watermark at every individual layer for detectability.

In this work, however, we identify a critical and previously underexplored limitation in this ``stronger-is-better'' approach. We demonstrate that while strong watermarks improve detectability within the current layer, they inevitably reduce the entropy of the token distribution for subsequent layers. Since the detectability of probabilistic watermarks fundamentally relies on the uncertainty (entropy) of the underlying distribution, this entropy reduction paradoxically weakens the effectiveness of watermarking in later layers. Consequently, we identify an inherent \textit{trade-off} between \textit{single-layer watermark strength} and the \textit{preservation of entropy} required for effective multi-layer ensembling, as illustrated in Fig.~\ref{fig:teaser}.
\begin{figure*}
    \centering
    \includegraphics[width=0.99\linewidth]{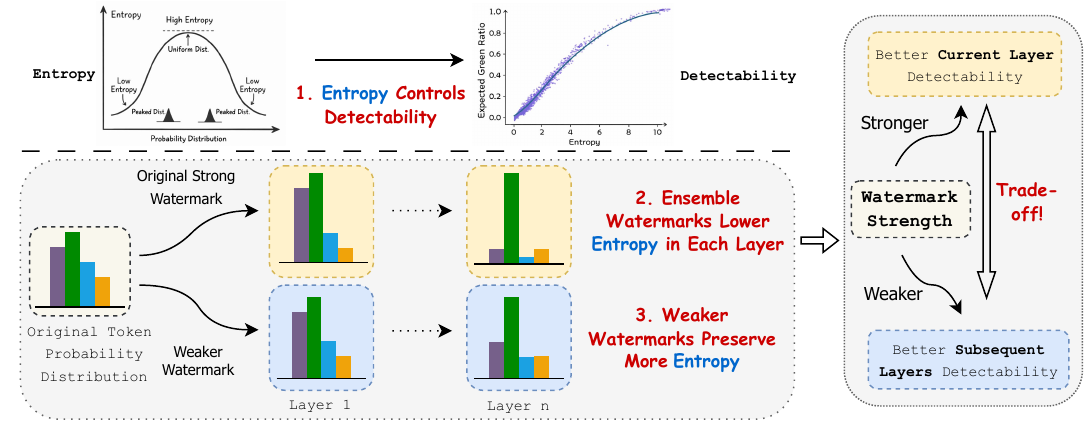}
    \vspace{-1mm}
    \caption{The relationship between entropy, watermark strength, and detectability in distortion-free watermark ensembles. Watermark detectability is closely tied to entropy. Stronger watermarks improve detectability within the current layer. However, they significantly reduce the entropy of the token distribution. In contrast, weaker watermarks preserve more entropy, thereby enhancing detectability in subsequent layers. We propose that there exists an inherent trade-off between the detectability across layers, which can be effectively controlled by adjusting the strength of the watermark.}
    \label{fig:teaser}
    \vspace{-4mm}
\end{figure*}

To resolve this dilemma, we propose a counter-intuitive yet theoretically grounded strategy: the use of \textbf{weaker single-layer watermarks} to improve overall ensemble performance. By interpolating between the original model distribution and the watermarked distribution using a mixing coefficient $\lambda$, we construct a general framework for weaker distortion-free watermarking. We theoretically prove that this approach mitigates the rapid decay of entropy and the expected green-list ratio across layers. By preserving the entropy of the distribution, weaker watermarks allow the ensemble to accumulate statistical evidence more efficiently over longer sequences, ultimately yielding superior detectability.

We summarize our contributions as follows:
\vspace{-1mm}
\begin{itemize}
    \item We demonstrate that for existing watermark methods, detectability is bounded by the entropy of the token distribution.
    \vspace{-2mm}
    \item We theoretically prove that during watermark ensemble generation, watermarks cause a monotonic decrease in both entropy and detectability across layers.
    \vspace{-2mm}
    \item We propose a novel framework to achieve better ensemble detectability by preserving more entropy per layer through \textit{weakening} the single-layer watermarking signal. We validate this approach empirically, showing it outperforms strong baselines in both detectability and robustness.
\end{itemize}

\section{Related Work}
\subsection{Distortion-Free Watermark}
To preserve the original output distribution of the language model, several studies have explored distortion-free watermarking techniques.
\citet{Aaronson2022} propose an early distortion-free method that applies the Gumbel-max trick to adjust token sampling, using prefix $n$-grams as watermark keys.
\citet{christ2023undetectable} employ inverse sampling to modify token distributions in a binary language model, with watermark keys determined by token positions.
ITS-edit and EXP-edit \citep{kuditipudi2023robust} apply inverse sampling and Gumbel-max, respectively, together with a fixed list of watermark keys.
\citet{hu2023unbiased} mix inverse sampling with a $\gamma$-reweighting strategy.
DiPmark \cite{wu2023dipmark} improves the $\gamma$-reweight approach and proposes a robust watermark detector.
STA-1 \cite{mao2024watermark} focuses on improving text quality in low-entropy generation settings.
\citet{dathathri2024scalable} introduce SynthID, which supports scalable distortion-free watermarking across multiple generations.
Finally, \citet{chen2025improved} propose MCmark, which utilizes distribution channels to improve the detectability while remaining distortion-free.

\vspace{-2mm}
\subsection{Watermark Ensemble}
\vspace{-1mm}
Recent research in watermarking for large language models has focused not only on individual watermarking schemes but also on combining multiple watermarking mechanisms to improve detectability and robustness. %One such direction is the development of ensemble methods that build stronger statistical signals by aggregating multiple independent watermark instances.
SynthID~\citep{dathathri2024scalable} uses multiple layers of random functions to strengthen watermark signals and outperforms simpler single-layer methods in detection performance. We further clarify our contributions compared to SynthID in App.~\ref{appendix:compare_synthid}. BiMark~\citep{feng2025bimark} achieves robust multi-bit watermarking by ensembling layers of bit-flip unbiased reweighting function based on random bipartitions of the vocabulary. \citet{wu2025ensemble} introduces an ensemble framework for distortion-free watermarking, which reduces the amount of text required for  detection by increasing the signal-to-noise ratio measured by statistical detectors.

\section{Distortion-Free Watermark Ensemble}

\subsection{Generation}

\paragraph{Notation.}
Let \( V \) denote the vocabulary with cardinality \( N = |V| \).
Given a prompt, a language model \( M \) generates tokens autoregressively.
At time step \( t \), the probability of generating the next token \( x_{t+1} \in V \),
conditioned on the preceding sequence \( \bm{x}_{1:t} \), is denoted by
\( P_M(x_{t+1} \mid \bm{x}_{1:t}) \in \mathcal{P} \),
where \( \mathcal{P} \) is the space of probability distributions over \( V \).

Watermarking is performed by reweighting the original distribution
\( P_M(\cdot \mid \bm{x}_{1:t}) \) into a watermarked distribution
\( F(P_M(\cdot \mid \bm{x}_{1:t}), k) \),
where \( F \in \mathcal{F} : \mathcal{P} \times \mathcal{K} \to \mathcal{P} \) is the watermarking function, \( k \in \mathcal{K} \) is a private watermark key sampled from a known distribution
\( P_{\mathcal{K}} \).
The key \( k \) is assumed to be exactly recoverable during detection.
Here, \( \mathcal{K} \) denotes the key space.

\vspace{-1mm}
\paragraph{Distortion-Free Watermark.}
Following prior work~\citep{hu2023unbiased,dathathri2024scalable},
a watermarking function $F$ is said to be \emph{distortion-free} if it preserves the original token distribution
in expectation over the key space.
Formally, for any \( P_M \in \mathcal{P} \) and any token \( x_{t+1} \in V \),
\begin{equation}\label{eq:unbiased-F}
\begin{aligned}
    &\mathbb{E}_{k \sim P_{\mathcal{K}}}
    \left[ F\!\left(P_M(x_{t+1} \mid \bm{x}_{1:t}),\, k\right) \right]\\
    =& P_M(x_{t+1} \mid \bm{x}_{1:t}).
\end{aligned}
\end{equation}

\vspace{-1mm}
\paragraph{Watermark Ensemble.}
Following~\citet{dathathri2024scalable,wu2026ensemble},
the detectability of a watermark can be further enhanced by ensembling,
which recursively applies the watermarking function \( F \) $n$ times using
independent watermark keys.
Specifically, given \( n \) i.i.d.\ keys \( \bm{k}_{1:n} \sim P_{\mathcal{K}}^n \),
the \( n \)-fold ensemble transform
\( \mathrm{ENS} : \mathbb{N} \times \mathcal{P} \times \mathcal{F} \times \mathcal{K}^n \to \mathcal{P} \)
is defined recursively as

\vspace{-3mm}
\begin{equation}\label{eq:ensemble-def}
\begin{aligned}
&\mathrm{ENS}(n, F, P(\cdot), \bm{k}_{1:n})\\
=&
\begin{cases}
F\!\left(
    \mathrm{ENS}(n-1, F, P(\cdot), \bm{k}_{1:n-1}),
    k_n
\right),\!\!&n \!>\! 1, \\[1pt]
F(P(\cdot), k_1),& n \!=\! 1 .
\end{cases}
\end{aligned}
\end{equation}

Following the terminology of~\citet{dathathri2024scalable},
each application of \( F \) corresponds to one \emph{watermark layer}.

If the watermarking function \( F \) is distortion-free and the keys
\( \bm{k}_{1:n} \) are sampled independently from \( P_{\mathcal{K}} \),
then the resulting \( n \)-layer watermark ensemble remains distortion-free
\citep{dathathri2024scalable,wu2026ensemble}.
That is, for any \( P_M(\cdot \mid \bm{x}_{1:t}) \in \mathcal{P} \),
\begin{equation}\label{eq:unbiased-ensemble}
\begin{aligned}
    &\mathbb{E}_{\bm{k}_{1:n} \sim P_{\mathcal{K}}^n}
    \left[ \mathrm{ENS}(n,F,P_M(\cdot \mid \bm{x}_{1:t}),\k_{1:n}) \right] \\
    =& P_M(\cdot \mid \bm{x}_{1:t}) .
\end{aligned}
\end{equation}

\subsection{Detection}

\paragraph{Single-Layer Detection.}
Most existing watermark detection methods~\citep{kirchenbauer2023watermark,dathathri2024scalable} adopt or resemble the red-green list framework introduced by~\citet{kirchenbauer2023watermark}. Given a private watermark key \(k\), the vocabulary \(V\) is randomly partitioned into a green list \(V_g^k\) and a red list \(V_r^k\), such that \(|V_g^k| = \gamma |V|\), where \(\gamma \in (0,1)\) is a fixed hyperparameter. During generation, the language model’s next-token distribution is modified to bias sampling toward green-list tokens and away from red-list tokens, and have:

\vspace{-1mm}
\begin{equation}
\begin{aligned}
    &\sum_{x_{t+1} \in V_g^k}F(P(x_{t+1}|\x_{1:t}),k)\\
    \ge & \sum_{x_{t+1} \in V_g^k}P(x_{t+1}|\x_{1:t})
\end{aligned}
\end{equation}

Detection is formulated as a hypothesis testing problem, with the null hypothesis \(H_0\) that the sequence was generated without watermarking, and the alternative hypothesis \(H_1\) that the sequence was generated with watermarking. Given a generated sequence \(\bm{x}_{1:T}\) and the corresponding recovered keys \(\bm{k}^{1:T}\), the detector computes the \emph{green ratio}
\vspace{-1mm}
\begin{equation}
    G = \frac{1}{T} \sum_{i=1}^T \mathbbm{1}\!\left(x_i \in V_g^{k^i}\right),
\end{equation}
i.e., the empirical fraction of tokens belonging to the green list. Under the assumption of independent token sampling, the corresponding z-score is computed as
\vspace{-1mm}
\begin{equation}
    z = \frac{(G - \gamma)T}{\sqrt{T\gamma(1-\gamma)}}.
\end{equation}
We refer readers to~\citet{kirchenbauer2023watermark} for a detailed derivation and analysis.

\vspace{-1mm}
\paragraph{Multi-Layer Detection.}
For an \(n\)-layer watermark ensemble, the single-layer detection framework can be naturally extended under the assumption that watermark keys across layers are independently and identically distributed. Let \(G_i\) denote the green ratio associated with the \(i\)-th layer:
\vspace{-1mm}
\begin{equation}
    G_i = \frac{1}{T} \sum_{j=1}^T \mathbbm{1}\!\left(x_j \in V_g^{k_i^j}\right), \quad i = 1, \ldots, n.
\end{equation}
Aggregating across all layers, the corresponding multi-layer z-score is given by
\begin{equation}
    z = \frac{\left(\sum_{i=1}^n G_i - \gamma n\right)T}{\sqrt{nT\gamma(1-\gamma)}}.
    \label{eq:multilayer-z-score}
\end{equation}

Finally, for a single-layer watermark $F$, we characterize the expected probability that the next generated token belongs to the green list under the original model distribution \(P_M(\cdot \mid \bm{x}_{1:t})\). We define the \emph{expected green ratio} function \(g_F : \mathcal{P} \rightarrow [0,1]\) as
\vspace{-1mm}
\begin{equation}
\begin{aligned}
    &g_{F}\!\left(P_M(\cdot \mid \bm{x}_{1:t})\right)\\
    =& \mathbb{E}_{k \sim P_{\mathcal{K}}}
    \sum_{x_{t+1} \in V_g^k}
    F\!\left(P_M(x_{t+1} \mid \bm{x}_{1:t}), k\right),
\end{aligned}
\end{equation}
where the expectation is taken over the watermark key distribution \(P_{\mathcal{K}}\).

\section{Rethinking Distortion-Free Ensemble Watermarks from an Entropy Perspective}

\subsection{Larger Entropy Leads to Better Detection Performance}

It is well established that the detectability of probabilistic watermarks depends critically on the uncertainty of the original token distribution prior to watermarking~\citep{hu2023unbiased,dathathri2024scalable,giboulot2024watermax}. Intuitively, when the model output distribution exhibits higher entropy, watermark-induced perturbations are less likely to be masked by deterministic biases of the base model, thereby yielding more reliable statistical signals for detection.

We formalize this intuition by analyzing the expected green-list ratio for each watermark method. Under the SynthID watermarking scheme~\citep{dathathri2024scalable}, we show that the expected green ratio under SynthID can be expressed as
\vspace{-1.5mm}
\begin{equation}
\begin{aligned}
     &g_\mathrm{SynthID}(P_M(\cdot \mid \bm{x}_{1:t}))\\
     =& \frac{3}{4}
     - \frac{1}{4} \sum_{x_{t+1}\in V}
     \big(P_M(x_{t+1} \mid \bm{x}_{1:t})\big)^2 \\
     \leq& \frac{3}{4}
     - \frac{1}{4}\exp\!\big(-H(P_M(\cdot \mid \bm{x}_{1:t}))\big),
\end{aligned}
\label{eq:synthid_expected_green_ratio}
\end{equation}
where
\(
H(P_M(\cdot \mid \bm{x}_{1:t}))
= -\sum_{x_{t+1}\in V}
P_M(x_{t+1} \mid \bm{x}_{1:t})
\log P_M(x_{t+1} \mid \bm{x}_{1:t})
\)
denotes the Shannon entropy of the model’s output distribution. A detailed derivation is provided in Appendix~\ref{appendix:expected_green}.

For other distortion-free ensemble watermarking schemes, deriving an analogous closed-form expression is generally intractable due to more complex reweighting mechanisms. Nevertheless, in Sec.~\ref{sec:entropy_vs_green}, we empirically demonstrate a consistent and strong positive correlation between the entropy of the base distribution and the expected green ratio across a variety of watermark designs, supporting the generality of this entropy-based perspective.

\subsection{Watermarks Reduce Entropy and Expected Green Ratio in Subsequent Layers}
\label{sec:wm_reduce_entropy}

We theoretically identify an important phenomenon: \emph{distortion-free watermarking inevitably reduces both the entropy of the token distribution and the expected green ratio in subsequent layers}. This effect compounds across layers, progressively weakening watermark detectability.

\begin{theorem}[Entropy Decrease under Distortion-Free Watermarking]
\label{thm:entropy-decreasing}
Let $F$ denote a distortion-free watermarking operator with private key $k \sim P_{\mathcal{K}}$. Then, in expectation over the watermark key, the Shannon entropy of the token distribution after watermarking does not increase:
\begin{equation}
\begin{aligned}
    &\mathbb{E}_{k \sim P_{\mathcal{K}}}
    \big[ H(F(P_M(\cdot \mid \bm{x}_{1:t}), k)) \big]\\
    \le&
    H(P_M(\cdot \mid \bm{x}_{1:t})) .
\end{aligned}
\end{equation}
\end{theorem}

\begin{proof}
Since Shannon entropy is a concave functional over the probability simplex, Jensen’s inequality yields
\begin{equation}
\begin{aligned}
    &H(P_M(\cdot \mid \bm{x}_{1:t}))\\
    =& H\!\left( \mathbb{E}_{k \sim P_{\mathcal{K}}}
        \big[ F(P_M(\cdot \mid \bm{x}_{1:t}), k) \big] \right) \\
    \ge&
    \mathbb{E}_{k \sim P_{\mathcal{K}}}
    \big[ H(F(P_M(\cdot \mid \bm{x}_{1:t}), k)) \big],
\end{aligned}
\end{equation}
which completes the proof.
\end{proof}

\begin{theorem}[Expected Green Ratio Decrease]
\label{thm:green-ratio-decreasing}
Under expectation over the watermark key, the expected green ratio does not increase after applying a distortion-free watermark:
\begin{equation}
\begin{aligned}
        &\mathbb{E}_{k \sim P_{\mathcal{K}}}
    \big[ g(F(P_M(\cdot \mid \bm{x}_{1:t}), k)) \big]\\
    \le&
    g(P_M(\cdot \mid \bm{x}_{1:t})) .
\end{aligned}
\end{equation}
\end{theorem}

\begin{proof}[Proof sketch]
The argument parallels Theorem~\ref{thm:entropy-decreasing}. If the expected green ratio functional $g(\cdot)$ is concave with respect to the token distribution, then Jensen’s inequality directly implies the result.

For SynthID, Eq.~\ref{eq:synthid_expected_green_ratio} admits a closed-form expression which is concave. For other representative distortion-free watermarking schemes, we also theoretically verify concavity of their expected green-ratio functions in Appendix~\ref{appendix:expected_green}.
\end{proof}

\paragraph{Implication.}
Taken together, Theorems~\ref{thm:entropy-decreasing} and~\ref{thm:green-ratio-decreasing} imply that watermark ensemble generation inevitably induces a monotonic reduction in both entropy and expected green ratio across successive layers. As a consequence, the statistical signal exploited by downstream detectors progressively diminishes, fundamentally limiting watermark detectability.

\subsection{A General Weaker Distortion-Free Watermark Framework}
\label{sec:weaker_df_wm}

A direct approach to mitigating the degradation of watermark detectability across layers is to reduce the entropy loss induced at each step, thereby maintaining a larger expected green ratio in subsequent layers. This motivates the design of a \emph{weaker} distortion-free watermark that perturbs the base token distribution more mildly.

Concretely, we construct a family of weaker watermarking functions by interpolating between the original model distribution and a distortion-free watermarked distribution. Let $F$ denote a distortion-free watermarking function. For a mixing coefficient $\lambda \in [0,1]$, we define
\vspace{-1mm}
\begin{equation}
\label{eq:weaker-watermark}
\begin{aligned}
&F_\lambda(P_M(\cdot \mid \bm{x}_{1:t}), k)\\
:=\;&
\lambda\, F(P_M(\cdot \mid \bm{x}_{1:t}), k)
+
(1-\lambda)\, P_M(\cdot \mid \bm{x}_{1:t}).
\end{aligned}
\end{equation}

The parameter $\lambda$ controls the watermark strength. $\lambda=1$ recovers the original watermark $F$, while smaller values of $\lambda$ yield distributions closer to the base model and thus induce smaller perturbations. We show that this interpolation preserves distortion-freeness and in expectation improves entropy relative to directly applying $F$, which in turn helps sustain detectability in later layers.

\begin{theorem}[Distortion-Freeness]
\label{thm:distortion-free}
If $F$ is distortion-free, then for any $\lambda \in [0,1]$, $F_\lambda$ is also distortion-free. In particular, for any token $x_{t+1}$,
\vspace{-1mm}
\begin{equation}
\begin{aligned}
    &\mathbb{E}_{k \sim P_{\mathcal{K}}}
\Big[
F_\lambda(P_M(x_{t+1} \mid \bm{x}_{1:t}), k)
\Big]\\
=&
P_M(x_{t+1} \mid \bm{x}_{1:t}).
\end{aligned}
\end{equation}
\end{theorem}
By linearity of expectation, the proof is straightforward, and we leave it in Appendix~\ref{appendix:weaker_wm}.

\begin{theorem}[Entropy Preservation]
\label{thm:entropy-preserve}
Assume $F$ is distortion-free. Then for any $\lambda \in [0,1]$,
\begin{equation}
\begin{aligned}
    &\mathbb{E}_{k \sim P_{\mathcal{K}}}
\Big[
H\!\big(F_\lambda(P_M(\cdot \mid \bm{x}_{1:t}), k)\big)
\Big]\\
\;\ge\;&
\mathbb{E}_{k \sim P_{\mathcal{K}}}
\Big[
H\!\big(F(P_M(\cdot \mid \bm{x}_{1:t}), k)\big)
\Big].
\end{aligned}
\end{equation}
\end{theorem}
\begin{proof}[Proof Sketch]
For each fixed key $k$, $F_\lambda(P_M(\cdot \mid \bm{x}_{1:t}),k)$ is a convex combination of
$F(P_M(\cdot \mid \bm{x}_{1:t}),k)$ and $P_M(\cdot \mid \bm{x}_{1:t})$.
Since Shannon entropy is concave on the probability simplex,
\[
H\!\big(F_\lambda(\cdot)\big)
\;\ge\;
\lambda\, H\!\big(F(\cdot)\big)
+
(1-\lambda)\, H\!\big(P_M(\cdot)\big).
\]
Taking expectation over $k$ and using distortion-freeness to relate $P_M$ to the key-average behavior yields the stated inequality. Full details are deferred to Appendix~\ref{appendix:weaker_wm}.
\end{proof}

The proposed interpolation in~\eqref{eq:weaker-watermark} offers a principled framework for modulating watermark strength without violating the distortion-free constraint. In this context, we observe a critical trade-off between immediate detectability and long-term sustainability: maximizing the watermark signal at the current step often collapses the entropy necessary for embedding signals in future steps. Our proposed weaker watermark, $F_\lambda$, directly addresses this tension. By reducing entropy degradation at the current step, $F_\lambda$ preserves a higher expected green ratio for subsequent layers, effectively counteracting the phenomenon of detectability decay.

\section{Experiments}

\begin{table*}[htbp]
\centering
\caption{Watermark detection performance with Llama-3.2-3B-Instruct and Mistral-7B-Instruct-v0.3 on the C4 dataset under different token budgets. We report TPR at fixed FPR thresholds (0.1\%, 0.01\%, 0.001\%) and the median detection p-value across samples. Lower p-values and higher TPR indicate stronger detectability.}
\vspace{-1mm}
\label{tab:main_exp}
\setlength{\tabcolsep}{2pt}
\resizebox{\textwidth}{!}{
\begin{tabular}{@{}ll|cccc|cccc@{}}
\toprule
 &  & \multicolumn{4}{c|}{150 tokens} & \multicolumn{4}{c}{250 tokens} \\ \midrule
 &  & \multicolumn{3}{c|}{TPR@FPR} & \multirow{2}{*}{\begin{tabular}[c]{@{}c@{}}Median\\ p-value\end{tabular}$\downarrow$} & \multicolumn{3}{c|}{TPR@FPR} & \multirow{2}{*}{\begin{tabular}[c]{@{}c@{}}Median\\ p-value\end{tabular}$\downarrow$} \\ \cmidrule(lr){3-5} \cmidrule(lr){7-9}
 &  & 0.1\%$\uparrow$ & 0.01\%$\uparrow$ & \multicolumn{1}{c|}{0.001\%$\uparrow$} &  & 0.1\%$\uparrow$ & 0.01\%$\uparrow$ & \multicolumn{1}{c|}{0.001\%$\uparrow$} &  \\ \midrule
\multicolumn{1}{l|}{\multirow{6}{*}{Llama3-3B}} & ENS-DiPmark ($\alpha$=0.5) & 52.8\% & 42.9\% & 35.2\% & 7.34e-4 & 74.0\% & 64.1\% & 54.9\% & 2.04e-6 \\
\multicolumn{1}{l|}{} & ENS-DiPmark ($\alpha$=0.4) & \textbf{55.5\%} & \textbf{45.2\%} & \textbf{38.2\%} & \textbf{4.20e-4} & \textbf{74.7\%} & \textbf{67.2\%} & \textbf{58.5\%} & \textbf{5.36e-7} \\ \cmidrule(l){2-10}
\multicolumn{1}{l|}{} & SynthID ($\lambda$=1) & 75.1\% & 68.3\% & 61.9\% & 1.33e-7 & 93.8\% & 88.4\% & 83.5\% & 1.91e-12 \\
\multicolumn{1}{l|}{} & SynthID ($\lambda$=0.8) & \textbf{79.7\%} & \textbf{74.5\%} & \textbf{67.1\%} & \textbf{1.57e-8} & \textbf{94.5\%} & \textbf{90.5\%} & \textbf{86.1\%} & \textbf{1.69e-13} \\ \cmidrule(l){2-10}
\multicolumn{1}{l|}{} & ENS-MCMark ($\lambda$=1) & 77.2\% & 71.4\% & 64.1\% & \textbf{7.80e-9} & 95.1\% & 91.8\% & 87.3\% & 8.75e-15 \\
\multicolumn{1}{l|}{} & ENS-MCMark ($\lambda$=0.8) & \textbf{78.5\%} & \textbf{71.6\%} & \textbf{65.2\%} & \textbf{7.80e-9} & \textbf{96.7\%} & \textbf{93.7\%} & \textbf{88.5\%} & \textbf{3.91e-15} \\ \midrule
\multicolumn{1}{l|}{\multirow{6}{*}{Mistral-7B}} & ENS-DiPmark ($\alpha$=0.5) & 22.6\% & 11.7\% & 6.3\% & 2.60e-2 & 46.8\% & 27.8\% & 19.1\% & 2.13e-3 \\
\multicolumn{1}{l|}{} & ENS-DiPmark ($\alpha$=0.4) & \textbf{25.4\%} & \textbf{13.7\%} & \textbf{6.4\%} & \textbf{1.37e-2} & \textbf{48.3\%} & \textbf{33.6\%} & \textbf{20.7\%} & \textbf{1.42e-3} \\ \cmidrule(l){2-10}
\multicolumn{1}{l|}{} & SynthID ($\lambda$=1) & 73.1\% & 58.8\% & 45.8\% & 2.06e-5 & 91.9\% & 83.2\% & 75.0\% & 3.23e-8 \\
\multicolumn{1}{l|}{} & SynthID ($\lambda$=0.8) & \textbf{78.8\%} & \textbf{64.3\%} & \textbf{49.4\%} & \textbf{1.06e-5} & \textbf{93.6\%} & \textbf{87.4\%} & \textbf{78.0\%} & \textbf{2.19e-8} \\ \cmidrule(l){2-10}
\multicolumn{1}{l|}{} & ENS-MCMark ($\lambda$=1) & 78.5\% & 67.1\% & 52.4\% & 9.11e-6 & 93.3\% & 87.5\% & 80.5\% & 7.32e-9 \\
\multicolumn{1}{l|}{} & ENS-MCMark ($\lambda$=0.8) & \textbf{82.2\%} & \textbf{72.2\%} & \textbf{57.6\%} & \textbf{2.50e-6} & \textbf{93.6\%} & \textbf{88.8\%} & \textbf{82.1\%} & \textbf{2.03e-9} \\ \bottomrule
\end{tabular}
}
\end{table*}

\begin{figure}[!h]

    \centering
    \includegraphics[width=\linewidth]{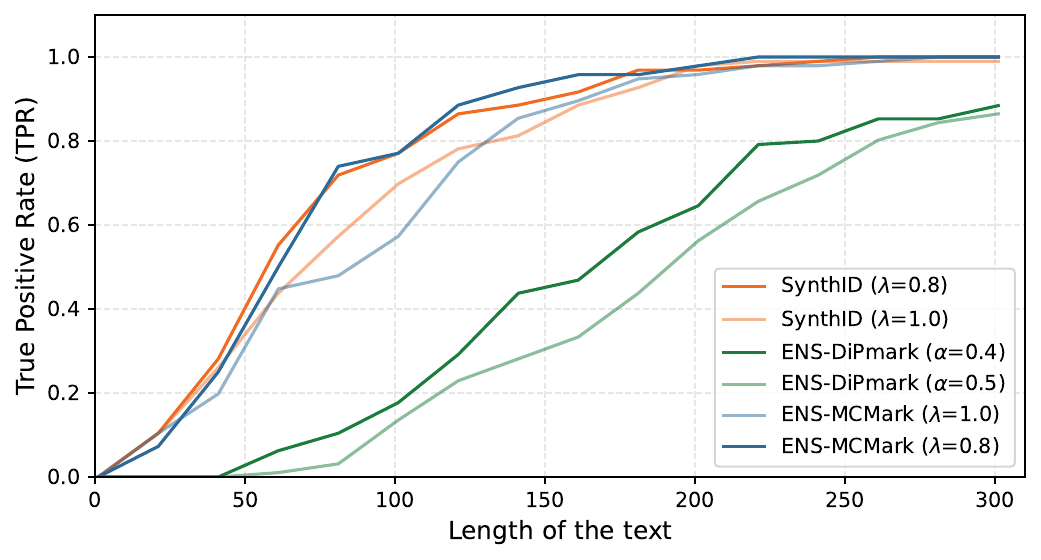}
    \caption{Watermark detectability as a function of text length on the MMW Story dataset using Llama-3.2-3B-Instruct. The threshold for false positive rate is set to 0.01\%. Weaker ensemble watermarks yield consistently higher detection performance.}
    \label{fig:res_vs_len}
    \vspace{-3mm}
\end{figure}

\subsection{Settings and Hyperparameters}

For distortion-free watermarking, we leverage SynthID~\citep{dathathri2024scalable}, ENS-DiPmark~\citep{wu2023dipmark,wu2026ensemble}, and ENS-MCMark~\citep{chen2025improved,wu2026ensemble}. We use the default hyperparameter settings for each method: for SynthID, we set $m=30$; for ENS-DiPmark, we set $n=5$; and for ENS-MCMark, we set $n=5$ and $l=20$, resulting in 30 layers for SynthID and 5 layers for both ENS-DiPmark and ENS-MCMark. To further explore the effect of the number of layers, we conduct an ablation study in Appendix~\ref{appendix:number_of_layers}.

Following prior works, we report the True Positive Rate (TPR) at specific controlled False Positive Rates (FPR), using values of 0.1\%, 0.01\%, and 0.001\% in our experiments. Additionally, we present the median p-values for the dataset.

Regarding watermark strength, ENS-DiPmark includes a hyperparameter $\alpha$ that controls the watermark strength. We use $\alpha = 0.5$ and $\alpha = 0.4$ for our experiments, as $\alpha = 0.5$ yields the best performance for single-layer DiPmark. For SynthID and ENS-MCMark, since no explicit hyperparameter is available to control watermark strength, we employ our proposed weaker watermark framework, as defined in Eq.~\eqref{eq:weaker-watermark}. In our main experiments, we test $\lambda = 1$ and $\lambda = 0.8$. Details of models and datasets used are shown in Appendix~\ref{appendix:models_datasets}.

\subsection{Experimental Results}

We report watermark detection performance across two backbone models and two generation lengths (150 and 250 tokens) in Tab.~\ref{tab:main_exp}. For distortion-free ensemble methods, weakening the watermark strength (reducing $\alpha$ in ENS-DiPmark or $\lambda$ in ENS-MCMark and SynthID) consistently improves detection performance, indicating that milder perturbations preserve more entropy and yield stronger downstream statistical evidence.
Fig.~\ref{fig:res_vs_len} illustrates how watermark detectability evolves as a function of text length. We observe that weaker watermark configurations consistently dominate their stronger counterparts, highlighting the benefit of preserving token entropy. Additional results on more datasets are provided in App.~\ref{appendix:res_vs_len} in Fig.~\ref{fig:res_vs_len_full}. These trends are consistent across models, datasets and watermark schemes, demonstrating the robustness and generality of the proposed weaker ensemble watermarking framework.

\begin{figure*}[!h]
    \centering
    \begin{subfigure}{0.32\textwidth}
        \centering
        \includegraphics[width=\linewidth]{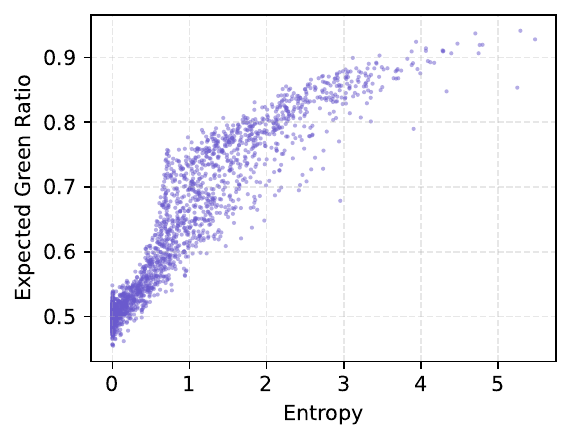}
        \caption{ENS-DiPmark}
        \label{fig:sub1}
    \end{subfigure}
    \hfill
    \begin{subfigure}{0.32\textwidth}
        \centering
        \includegraphics[width=\linewidth]{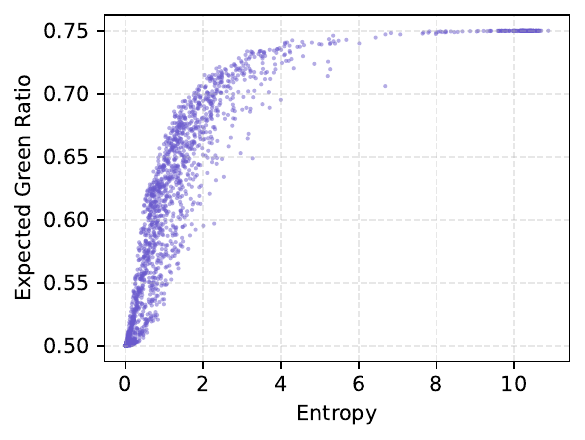}
        \caption{SynthID}
        \label{fig:sub2}
    \end{subfigure}
    \hfill
    \begin{subfigure}{0.32\textwidth}
        \centering
        \includegraphics[width=\linewidth]{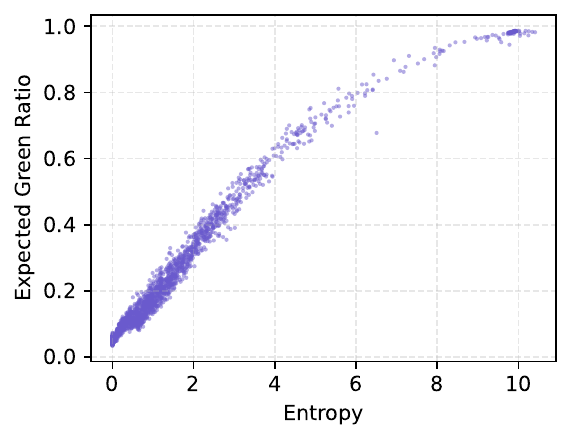}
        \caption{ENS-MCMark}
        \label{fig:sub3}
    \end{subfigure}
    \caption{Correlation between token-distribution entropy and expected green ratio. Results are computed on the MMW Story dataset with Llama3.2-3B-Instruct using 2000 randomly sampled tokens, showing a strong positive association between entropy and expected green ratio.}
    \vspace{-3mm}
    \label{fig:entropy_vs_expected_green_ratio}
\end{figure*}

\section{Analysis}
\begin{figure}[!h]
    \centering
    \begin{subfigure}{0.95\linewidth}
        \centering
        \includegraphics[width=\linewidth]{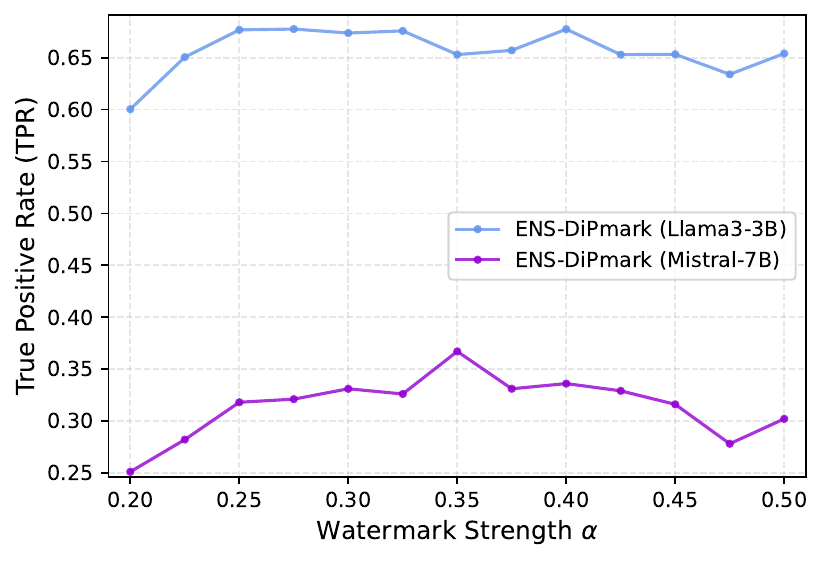}
        \caption{ENS-DiPmark}
        \label{fig:sub1}
    \end{subfigure}
    \\
    \begin{subfigure}{0.95\linewidth}
        \centering
        \includegraphics[width=\linewidth]{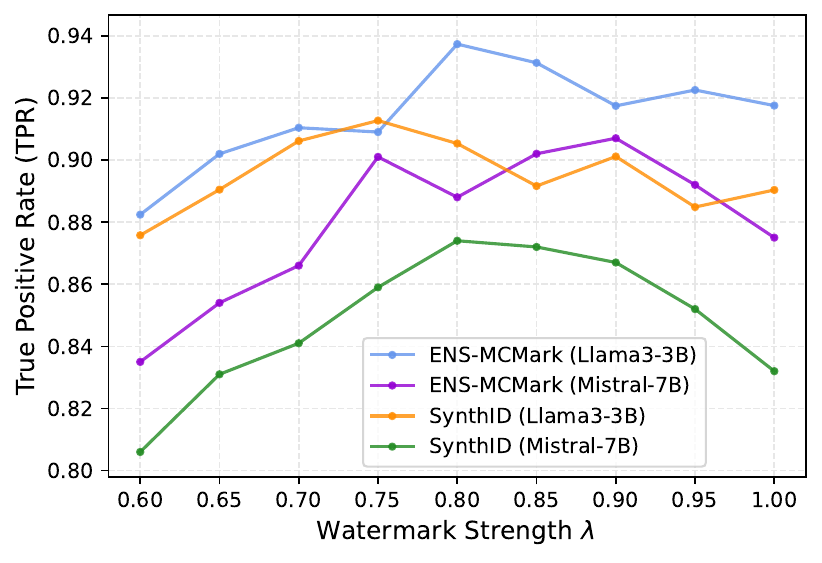}
        \caption{SynthID and ENS-MCMark}
        \label{fig:sub2}
    \end{subfigure}
    \vspace{-1mm}
    \caption{Ablation of watermark strength on detection performance on the C4 dataset with false positive rate set to 0.01\% and token length set to 250. True positive rate (TPR) is reported for varying strength parameters $\alpha$ (ENS-DiPmark) and $\lambda$ (SynthID, ENS-MCMark) on Llama3-3B and Mistral-7B. Moderate weakening consistently yields the best detectability.}
    \vspace{-3mm}
    \label{fig:ablation}
\end{figure}

\subsection{Effect of Watermark Strength}

We study the effect of watermark strength on detection performance in Fig.~\ref{fig:ablation}. Decreasing the watermark strength generally improves the true positive rate, with peak performance attained around $\lambda$=0.8 and $\alpha$=0.4. This behavior indicates that excessively strong perturbations can be detrimental, as they induce excessive entropy reduction that undermines downstream detectability. Importantly, the strength parameters $\alpha$ and $\lambda$ explicitly govern the trade-off between detectability at the current layer and entropy preservation, which in turn affects detectability in subsequent layers. These results provide empirical evidence supporting our central claim that weaker distortion-free watermarking achieves a more favorable balance between immediate detection strength and overall detectability.

\begin{figure*}[!h]
    \centering
    \begin{subfigure}{0.31\textwidth}
        \centering
        \includegraphics[width=\linewidth]{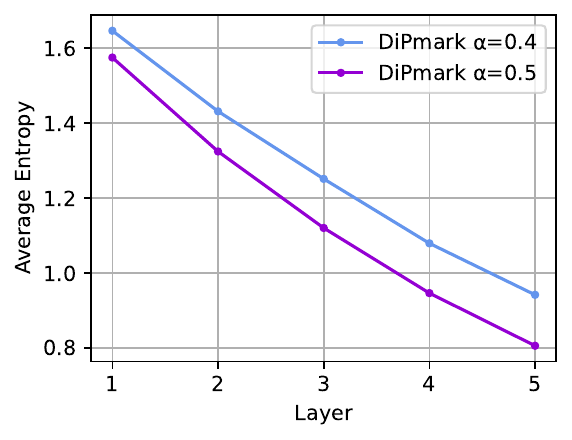}
        \caption{ENS-DiPmark}
        \label{fig:sub1}
    \end{subfigure}
    \hfill
    \begin{subfigure}{0.31\textwidth}
        \centering
        \includegraphics[width=\linewidth]{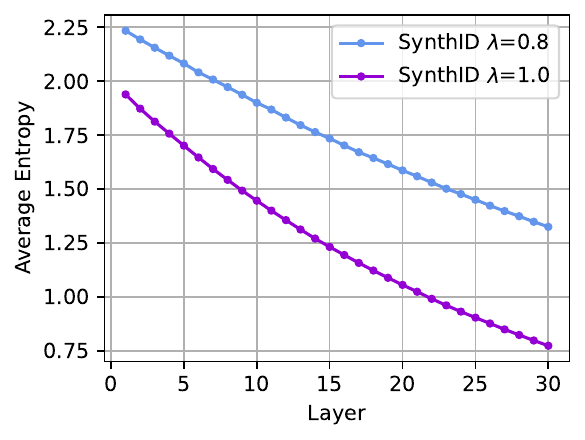}
        \caption{SynthID}
        \label{fig:sub2}
    \end{subfigure}
    \hfill
    \begin{subfigure}{0.31\textwidth}
        \centering
        \includegraphics[width=\linewidth]{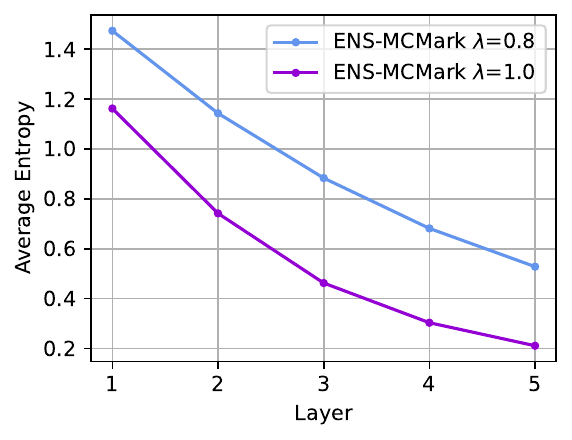}
        \caption{ENS-MCMark}
        \label{fig:sub3}
    \end{subfigure}
    \vspace{-1mm}
    \caption{Average token entropy per layer under different watermark strengths on the C4 dataset using Llama-3.2-3B-Instruct with sequence length fixed to 150. Weaker watermark settings consistently preserve higher entropy across layers.}
    \label{fig:entropy_per_layer}
    \vspace{-2mm}
\end{figure*}

\begin{figure*}[!h]
    \centering
    \begin{subfigure}{0.31\textwidth}
        \centering
        \includegraphics[width=\linewidth]{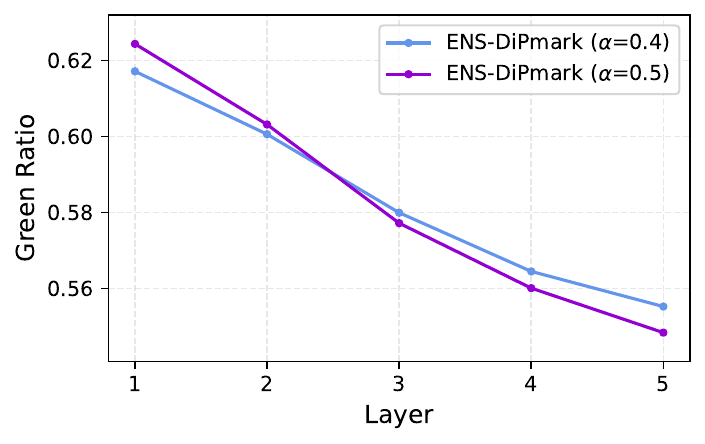}
        \caption{ENS-DiPmark}
        \label{fig:sub1}
    \end{subfigure}
    \hfill
    \begin{subfigure}{0.31\textwidth}
        \centering
        \includegraphics[width=\linewidth]{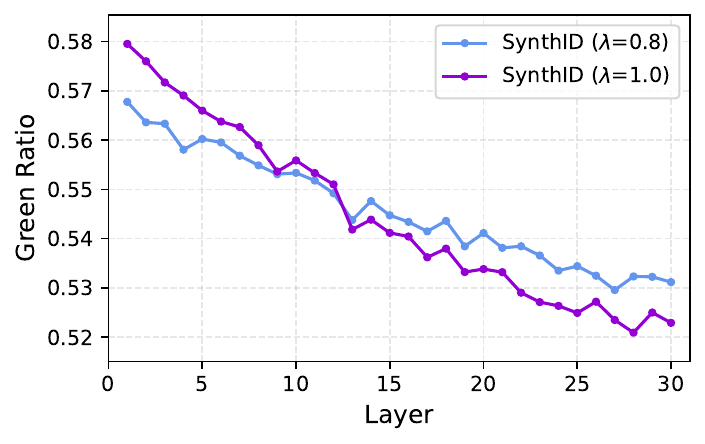}
        \caption{SynthID}
        \label{fig:sub2}
    \end{subfigure}
    \hfill
    \begin{subfigure}{0.31\textwidth}
        \centering
        \includegraphics[width=\linewidth]{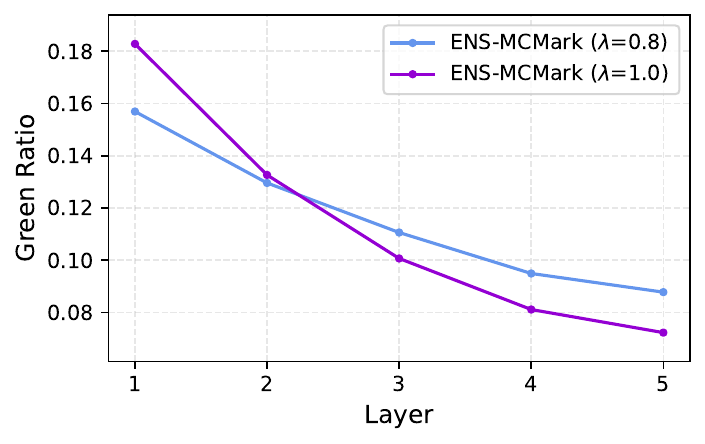}
        \caption{ENS-MCMark}
        \label{fig:sub3}
    \end{subfigure}
    \caption{Actual green ratio $G_i$ per layer under different watermark strengths on the C4 dataset using Llama-3.2-3B-Instruct with sequence length fixed to 150. Higher entropy preservation from weaker watermarking leads to larger green ratios in later layers.}
    \vspace{-2mm}
    \label{fig:green_ratio_per_layer}
\end{figure*}
\subsection{Correlation between Entropy and Expected Green Ratio}
\label{sec:entropy_vs_green}

\begin{table*}[!h]
\centering
\caption{Robustness of watermark detection under four attacks with different watermark strength. We report TPR (at FPR=0.1\%) and median p-value on the C4 dataset using Llama3.2-3B-Instruct with sequence length 250.}
\label{tab:robustness}
\setlength{\tabcolsep}{3pt}
\begin{tabular}{@{}l|cc|cc|cc|cc@{}}
\toprule
 & \multicolumn{2}{c|}{\begin{tabular}[c]{@{}c@{}}Random Token\\ Replacement\end{tabular}} & \multicolumn{2}{c|}{\begin{tabular}[c]{@{}c@{}}GPT\\ Back Translation\end{tabular}} & \multicolumn{2}{c|}{\begin{tabular}[c]{@{}c@{}}GPT\\ Rephrase\end{tabular}} & \multicolumn{2}{c}{DIPPER} \\ \midrule
 & TPR$\uparrow$ & \begin{tabular}[c]{@{}c@{}}Median\\ p-value\end{tabular}$\downarrow$ & TPR$\uparrow$ & \begin{tabular}[c]{@{}c@{}}Median\\ p-value\end{tabular}$\downarrow$ & TPR$\uparrow$ & \begin{tabular}[c]{@{}c@{}}Median\\ p-value\end{tabular}$\downarrow$ & TPR$\uparrow$ & \begin{tabular}[c]{@{}c@{}}Median\\ p-value\end{tabular}$\downarrow$ \\ \midrule
ENS-DiPmark ($\alpha$=0.5) & 54.2\% & 4.92e-4 & 18.6\% & 6.79e-2 & 1.9\% & 5.27e-1 & 1.6\% & 6.30e-1 \\
ENS-DiPmark ($\alpha$=0.4) & \textbf{55.9\%} & \textbf{3.14e-4} & \textbf{22.4\%} & \textbf{3.92e-2} & \textbf{2.5\%} & \textbf{4.61e-1} & \textbf{2.0\%} & \textbf{5.95e-1} \\ \midrule
SynthID ($\lambda$=1) & 81.1\% & 9.48e-8 & 51.8\% & 6.90e-4 & 12.7\% & 6.82e-2 & 8.7\% & 9.59e-2 \\
SynthID ($\lambda$=0.8) & \textbf{84.8\%} & \textbf{3.44e-8} & \textbf{54.2\%} & \textbf{4.99e-4} & \textbf{13.0\%} & \textbf{6.52e-2} & \textbf{10.3\%} & \textbf{9.21e-2} \\ \midrule
ENS-MCMark ($\lambda$=1) & 85.9\% & 73.2e-9 & 55.3\% & 2.97e-4 & 13.8\% & 6.27e-2 & 12.3\% & 9.92e-2 \\
ENS-MCMark ($\lambda$=0.8) & \textbf{87.5\%} & \textbf{3.88e-9} & \textbf{58.9\%} & \textbf{1.24e-4} & \textbf{16.4\%} & \textbf{3.78e-2} & \textbf{12.7\%} & \textbf{7.93e-2} \\ \bottomrule
\end{tabular}
\end{table*}

We visualize the empirical relationship between the entropy of the pre-watermark token distribution and the resulting expected green ratio for three watermarking schemes in Fig.~\ref{fig:entropy_vs_expected_green_ratio}. Across ENS-DiPmark, SynthID, and ENS-MCMark, we observe a strong positive association: higher-entropy distributions yield systematically larger expected green ratios, whereas low-entropy distributions compress the green ratio toward its baseline level, reducing the statistical margin available for detection. The results provide empirical evidence for our entropy-based perspective: entropy serves as a key determinant of watermark detectability through its control of the expected green ratio.

\subsection{Entropy and Green Ratio in Each Layer}

 We empirically analyze how watermark strength influences the evolution of entropy and green ratio across layers. Fig.~\ref{fig:entropy_per_layer} shows that, for all methods, the token distribution entropy decreases monotonically for each layer. Importantly, weaker watermark settings (smaller $\alpha$ and $\lambda$) consistently preserve higher entropy at each layer, slowing entropy collapse for all the layers. Fig.~\ref{fig:green_ratio_per_layer} further demonstrates that this entropy preservation directly translates into higher expected green ratios in later layers, while stronger settings suffer from accelerated degradation. Together, these results provide empirical support for our theoretical analysis in Sec.~\ref{sec:wm_reduce_entropy}, confirming that watermark strength controls a fundamental trade-off between immediate detectability and long-horizon statistical capacity, and that weaker distortion-free watermarking better preserves entropy and green ratio, thereby enabling stronger detectability in subsequent layers.

\subsection{Robustness Experiments}

We evaluate the robustness of our proposed weaker watermark framework in Tab.~\ref{tab:robustness} under four attacks, including random token replacement, GPT-based back translation, GPT rephrasing~\citep{sadasivan2023can}, and DIPPER~\citep{krishna2023paraphrasing} paraphrasing attack, following the settings of~\citet{wu2026ensemble}. The results further demonstrate that our weaker distortion-free watermark framework does not sacrifice robustness. Instead, it consistently improves detectability under attacks. Across all four perturbations, the weaker configurations achieve higher TPR and smaller median p-values, indicating stronger statistical evidence.

\section{Conclusion}
In this work, we identify a critical limitation in existing watermark ensembles where maximizing single-layer strength paradoxically degrades long-term detectability by reducing the entropy of the token distribution. To address this trade-off, we propose a general framework that utilizes weaker distortion-free watermarks to preserve entropy across layers, thereby maintaining a higher expected green ratio for subsequent layers. Our theoretical proofs and empirical evaluations confirm that this entropy-preserving strategy significantly outperforms strong baselines, achieving superior detectability and robustness in multi-layer settings.

\section*{Limitations}

While our work demonstrates the effectiveness of weaker distortion-free watermarking ensembles, there are a few limitations to consider. First, our experimental evaluation was only conducted on open-source models with sizes ranging from 3B to 7B parameters.

Additionally, we did not explore adaptive strategies that dynamically tune the watermark strength based on the real-time entropy of the token distribution, which could potentially further optimize the trade-off between detectability and signal preservation.

\section*{Acknowledgments}

This work was partially supported by NSF IIS 2347592, 2348169, DBI 2405416, CCF 2348306, CNS 2347617, RISE 2536663.

\bibliography{custom}

\appendix

\section{Expected Green Ratio for Single-Layer Distortion-Free Watermarks}
\label{appendix:expected_green}
In this section, we provide the theoretical proofs for the concavity of the expected green ratio function $g(\cdot)$ and its relationship with entropy for the discussed watermarking schemes. For brevity, we denote $P_M(\cdot|x_{1:t})$ as $P(\cdot)$ when the context is unambiguous.

\subsection{SynthID}

\paragraph{Expected Green Ratio.} As established in prior work~\citep{dathathri2024scalable}, the expected green ratio for the SynthID scheme is given by:\[g_{\text{SynthID}}(P) = \frac{3}{4} - \frac{1}{4} \sum_{x \in V} (P(x))^2.\]

\paragraph{Proof of Concavity.} We aim to show that $g_{\text{SynthID}}(P)$ is a concave function over the probability simplex $\mathcal{P}$. Let $f(P) = \sum_{x \in V} (P(x))^2$. This represents the collision probability (or the squared $L_2$ norm) of the distribution. Consider the function $h(y) = y^2$, which is strictly convex for $y \in \mathbb{R}$. Since $f(P)$ is a sum of convex functions applied to the components of $P$, $f(P)$ is convex with respect to $P$.The expected green ratio is an affine transformation of $f(P)$ with a negative coefficient:

$$g_{\text{SynthID}}(P) = \frac{3}{4} - \frac{1}{4} f(P).$$

Since multiplying a convex function by a negative constant reverses the inequality, $- \frac{1}{4} f(P)$ is concave. Adding a constant ($\frac{3}{4}$) preserves concavity. Therefore, $g_{\text{SynthID}}(P)$ is concave.

\paragraph{Relation with Entropy.} (Proof of Eq. 9) We provide the derivation for the bound:$$g_{\text{SynthID}}(P) \le \frac{3}{4} - \frac{1}{4} \exp(-H(P)).$$

Recall the definition of Shannon entropy $H(P) = - \sum_{x \in V} P(x) \ln P(x)$.We utilize the relationship between the collision probability and entropy. By Jensen's inequality applied to the concave function $\ln(\cdot)$:

\begin{equation}
\begin{aligned}
    &\ln\left( \sum_{x \in V} P(x)^2 \right)\\
    =& \ln \left( \mathbb{E}_{x \sim P} [P(x)] \right) \\
    \ge &\mathbb{E}_{x \sim P} [\ln P(x)].
\end{aligned}
\end{equation}

Substituting the definition of entropy:

$$\mathbb{E}_{x \sim P} [\ln P(x)] = \sum_{x \in V} P(x) \ln P(x) = -H(P).$$

Thus, we have:

$$\ln\left( \sum_{x \in V} P(x)^2 \right) \ge -H(P).$$

Exponentiating both sides (since $e^x$ is monotonically increasing):

$$\sum_{x \in V} P(x)^2 \ge \exp(-H(P)).$$

Now, substituting this inequality back into the expected green ratio equation:

$$g_{\text{SynthID}}(P) = \frac{3}{4} - \frac{1}{4} \sum_{x \in V} (P(x))^2.$$

Multiplying the inequality by $-\frac{1}{4}$ reverses the sign:

$$-\frac{1}{4} \sum_{x \in V} (P(x))^2 \le -\frac{1}{4} \exp(-H(P)).$$

Finally, adding $\frac{3}{4}$ to both sides yields:

$$g_{\text{SynthID}}(P) \le \frac{3}{4} - \frac{1}{4} \exp(-H(P)).$$

This confirms that the expected green ratio is upper-bounded by a function of the entropy.

\subsection{DiPmark}

\begin{table*}[]
\centering
\caption{Ablation study on the number of ensemble layers for ENS-DiPmark using the Llama3-3B and Mistral-7B models. Results are reported as the true positive rate at a false positive rate of 0.01\% with 250 tokens per sample on the C4 dataset.}
\label{tab:layers_dipmark}
\begin{tabular}{@{}l|l|cccccc@{}}
\toprule
\multirow{2}{*}{Model} & \multirow{2}{*}{Watermark} & \multicolumn{6}{c}{Layer Number} \\ \cmidrule(l){3-8}
 &  & 1 & 2 & 3 & 4 & \multicolumn{1}{c|}{5} & Best \\ \midrule
\multirow{2}{*}{Llama3-3B} & ENS-DiPmark ($\alpha$=0.5) & \textbf{40.4\%} & \textbf{59.0\%} & 63.7\% & 63.9\% & \multicolumn{1}{c|}{64.1\%} & 64.1\% \\
 & ENS-DiPmark ($\alpha$=0.4) & 38.57\% & 57.2\% & \textbf{63.9\%} & \textbf{66.9\%} & \multicolumn{1}{c|}{\textbf{67.2\%}} & \textbf{67.2\%} \\ \midrule
\multirow{2}{*}{Mistral-7B} & ENS-DiPmark ($\alpha$=0.5) & \textbf{11.1\%} & \textbf{29.5\%} & 30.4\% & 30.7\% & \multicolumn{1}{c|}{27.8\%} & 30.7\% \\
 & ENS-DiPmark ($\alpha$=0.4) & 9.0\% & 28.2\% & \textbf{33.0\%} & \textbf{35.5\%} & \multicolumn{1}{c|}{\textbf{33.6\%}} & \textbf{35.5\%} \\ \bottomrule
\end{tabular}
\end{table*}

\textbf{\begin{table*}[]
\centering
\caption{Ablation study on the number of ensemble layers for SynthID using the Llama3-3B and Mistral-7B models. Results are reported as the true positive rate at a false positive rate of 0.01\% with 250 tokens per sample on the C4 dataset.}
\label{tab:layers_synthid}
\begin{tabular}{@{}l|c|ccccccc@{}}
\toprule
\multirow{2}{*}{Model} & \multicolumn{1}{l|}{\multirow{2}{*}{Watermark}} & \multicolumn{7}{c}{Layer Number} \\ \cmidrule(l){3-9}
 & \multicolumn{1}{l|}{} & 5 & 10 & 15 & 20 & 25 & \multicolumn{1}{c|}{30} & Best \\ \midrule
\multirow{2}{*}{Llama3-3B} & SynthID ($\lambda$=1) & \textbf{72.2\%} & \textbf{85.4\%} & \textbf{88.0\%} & 88.2\% & 88.6\% & \multicolumn{1}{c|}{88.4\%} & 88.6\% \\
 & SynthID ($\lambda$=0.8) & 60.5\% & 79.8\% & 85.8\% & \textbf{89.1\%} & \textbf{90.4\%} & \multicolumn{1}{c|}{\textbf{90.5\%}} & \textbf{90.5\%} \\ \midrule
\multirow{2}{*}{Mistral-7B} & SynthID ($\lambda$=1) & \textbf{59.0\%} & \textbf{81.4\%} & \textbf{85.2\%} & \textbf{86.0\%} & \textbf{85.5\%} & \multicolumn{1}{c|}{83.2\%} & 86.0\% \\
 & SynthID ($\lambda$=0.8) & 36.1\% & 69.6\% & 78.7\% & 82.9\% & 85.3\% & \multicolumn{1}{c|}{\textbf{87.4\%}} & \textbf{87.4\%} \\ \bottomrule
\end{tabular}
\end{table*}}

\paragraph{Expected Green Ratio.} For DiPmark~\citep{wu2023dipmark}, the expected green ratio is:
$$g_{\text{DiPmark}}(P) = \mathbb{E}_{g} \left[ \min\left(1, \sum_{i=1}^{|V|} P(x_i)g_i+\alpha\right) \right],$$

where $g_i \in \{0,1\}$ determines the partition of the vocabulary, $\sum_i g_i=\frac{|V|}{2}$ and $\alpha$ is a shift parameter.

\paragraph{Proof of Concavity.} We show that $g_{\text{DiPmark}}(P)$ is concave with respect to $P$.

Let $L(P, g) = \sum_{i=1}^{|V|} P(x_i)g_i + \alpha$. Note that for a fixed key/partition $g$, $L(P, g)$ is a linear function of $P$. Define the function $h(y) = \min(1, y)$. The function $h(y)$ is the minimum of two concave functions. Thus, $h(y)$ is concave. Since the composition of a concave function with a linear function is concave, the term inside the expectation:

\begin{equation}
    \begin{aligned}
        &\phi(P, g)= h(L(P, g))\\
        =& \min\left(1, \sum_{i=1}^{|V|} P(x_i)g_i + \alpha\right),
    \end{aligned}
\end{equation}

is concave with respect to $P$ for any fixed $g$. The expected green ratio $g_{\text{DiPmark}}(P)$ is the expectation of $\phi(P, g)$ over the distribution of watermark keys $g$. Since expectation is a linear operator with non-negative weights (probabilities), and a non-negative weighted sum of concave functions preserves concavity, $g_{\text{DiPmark}}(P)$ is concave.

\subsection{MCMark}
\paragraph{Expected Green Ratio.} For MCMark~\citep{chen2025improved}, with channel number $l$, the expected green ratio is:

$$g_{\text{MCMark}}(P) = \mathbb{E}_{g} \left[ \min\left(1, l \sum_{i=1}^{|V|} P(x_i)g_i\right) \right],$$

where $g_i \in \{0,1\}$ determines the partition of the vocabulary, $\sum_i g_i=\frac{|V|}{l}$.

\paragraph{Proof of Concavity.} The proof follows the same logic as DiPmark.

Let $K(P, g) = l \sum_{i=1}^{|V|} P(x_i)g_i$. This term represents the scaled probability mass falling into a specific channel determined by $g$. $K(P, g)$ is strictly linear with respect to the distribution $P$. Consider the function $h(y) = \min(1, y)$, which is a concave function. Therefore, the composition $\psi(P, g) = h(K(P, g)) = \min(1, l \sum P(x_i)g_i)$ is concave with respect to $P$.

Finally, taking the expectation over $g$:

$$g_{\text{MCMark}}(P) = \mathbb{E}_{g} [\psi(P, g)].$$

Since the expectation preserves concavity, $g_{\text{MCMark}}(P)$ is a concave function of the token distribution $P$.

\begin{figure*}[!h]
    \centering
    \begin{subfigure}{0.49\textwidth}
        \centering
        \includegraphics[width=\linewidth]{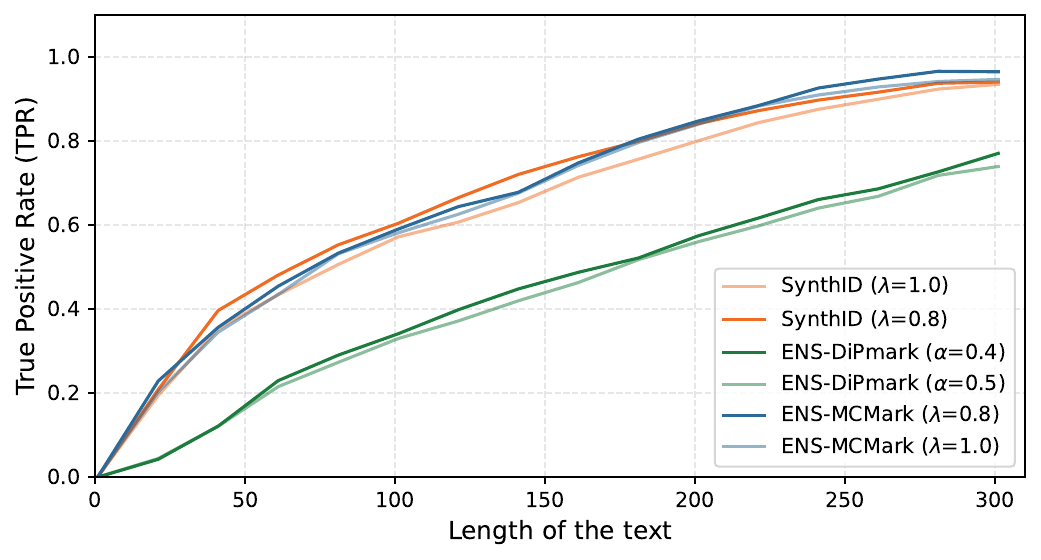}
        \caption{C4}
        \label{fig:sub1}
    \end{subfigure}
    \hfill
    \begin{subfigure}{0.49\textwidth}
        \centering
        \includegraphics[width=\linewidth]{figures/res_vs_len/final_mmw_story/combined_acc_vs_len_3files_alpha-1.0_l-1.pdf}
        \caption{MMW Story}
        \label{fig:sub2}
    \end{subfigure}
    \vspace{-2mm}
    \caption{Watermark detectability as a function of text length on C4 and Longform QA using Llama-3.2-3B-Instruct. The threshold for false positive rate is set to 0.01\%. Weaker ensemble watermarks yield consistently higher detection performance across datasets.}
    \label{fig:res_vs_len_full}
    \vspace{-3mm}
\end{figure*}
\section{A General Weaker Distortion-Free Watermark Framework}
\label{appendix:weaker_wm}

\subsection{Proof for Theorem~\ref{thm:distortion-free}}
\begin{proof}
By definition of $F_\lambda$ and linearity of expectation, we have
\begin{equation}
\begin{aligned}
&\mathbb{E}_{k \sim P_{\mathcal{K}}}
\big[
F_\lambda(P_M(\cdot \mid \bm{x}_{1:t}), k)
\big]\\
=&
\lambda
\mathbb{E}_{k \sim P_{\mathcal{K}}}
\big[
F(P_M(\cdot \mid \bm{x}_{1:t}), k)
\big]\\
&+ (1-\lambda) P_M(\cdot \mid \bm{x}_{1:t}).
\end{aligned}
\end{equation}
Since $F$ is distortion-free by assumption,
\(
\mathbb{E}_{k \sim P_{\mathcal{K}}}
[
F(P_M(\cdot \mid \bm{x}_{1:t}), k)
]
=
P_M(\cdot \mid \bm{x}_{1:t}),
\)
which immediately yields
\(
\mathbb{E}_{k \sim P_{\mathcal{K}}}
[
F_\lambda(P_M(\cdot \mid \bm{x}_{1:t}), k)
]
=
P_M(\cdot \mid \bm{x}_{1:t}).
\)
\end{proof}

\begin{table*}[]
\centering
\caption{Ablation study on the number of ensemble layers for ENS-MCMark using the Llama3-3B and Mistral-7B models. Results are reported as the true positive rate at a false positive rate of 0.01\% with 250 tokens per sample on the C4 dataset.}
\label{tab:layers_mcmark}
\begin{tabular}{@{}l|c|cccccc@{}}
\toprule
\multirow{2}{*}{Model} & \multicolumn{1}{l|}{\multirow{2}{*}{Watermark}} & \multicolumn{6}{c}{Layer Number} \\ \cmidrule(l){3-8}
 & \multicolumn{1}{l|}{} & 1 & 2 & 3 & 4 & \multicolumn{1}{c|}{5} & Best \\ \midrule
\multirow{2}{*}{Llama3-3B} & ENS-MCMark ($\lambda$=1) & \textbf{86.6\%} & \textbf{91.1\%} & 92.0\% & 92.4\% & \multicolumn{1}{c|}{91.8\%} & 92.4\% \\
 & ENS-MCMark ($\lambda$=0.8) & 77.1\% & 87.8\% & \textbf{92.6\%} & \textbf{93.6\%} & \multicolumn{1}{c|}{\textbf{93.7\%}} & \textbf{93.7\%} \\ \midrule
\multirow{2}{*}{Mistral-7B} & ENS-MCMark ($\lambda$=1) & \textbf{78.4\%} & \textbf{86.6\%} & \textbf{88.1\%} & \textbf{87.5\%} & \multicolumn{1}{c|}{87.5\%} & 88.1\% \\
 & ENS-MCMark ($\lambda$=0.8) & 57.8\% & 79.2\% & 85.0\% & 86.2\% & \multicolumn{1}{c|}{\textbf{88.8\%}} & \textbf{88.8\%} \\ \bottomrule
\end{tabular}
\end{table*}

\subsection{Proof for Theorem~\ref{thm:entropy-preserve}}

\begin{proof}
Fix a key $k$. By construction,
\(
F_\lambda(P_M(\cdot \mid \bm{x}_{1:t}), k)
\)
is a convex combination of the two distributions
\(
F(P_M(\cdot \mid \bm{x}_{1:t}), k)
\)
and
\(
P_M(\cdot \mid \bm{x}_{1:t}).
\)
Since Shannon entropy is a concave function over the probability simplex, we have
\begin{equation}
\begin{aligned}
&H(F_\lambda(P_M(\cdot \mid \bm{x}_{1:t}), k))\\
=&
H\big(
\lambda F(P_M(\cdot \mid \bm{x}_{1:t}), k)
+ (1-\lambda) P_M(\cdot \mid \bm{x}_{1:t})
\big) \\
\ge&
\lambda H(F(P_M(\cdot \mid \bm{x}_{1:t}), k))\\
&+ (1-\lambda) H(P_M(\cdot \mid \bm{x}_{1:t})).
\end{aligned}
\end{equation}
Taking expectation over $k \sim P_{\mathcal{K}}$ yields
\begin{equation}
\begin{aligned}
&\mathbb{E}_{k}
\big[
H(F_\lambda(P_M(\cdot \mid \bm{x}_{1:t}), k))
\big]\\
\ge&
\lambda
\mathbb{E}_{k}
\big[
H(F(P_M(\cdot \mid \bm{x}_{1:t}), k))
\big]\\
&+ (1-\lambda) H(P_M(\cdot \mid \bm{x}_{1:t})).
\end{aligned}
\end{equation}
Finally, since $F$ is distortion-free, Jensen's inequality implies
\(
\mathbb{E}_{k}[H(F(P_M(\cdot \mid \bm{x}_{1:t}), k))]
\le H(P_M(\cdot \mid \bm{x}_{1:t})).
\)
Substituting this bound into the inequality above gives
\(
\mathbb{E}_{k}[H(F_\lambda(\cdot))]
\ge
\mathbb{E}_{k}[H(F(\cdot))],
\)
which completes the proof.
\end{proof}

\section{Models and Datasets}
\label{appendix:models_datasets}
In our experiments, we follow prior works~\citep{chen2025improved,wu2026ensemble} and utilize several datasets, including a subset of the C4 dataset~\citep{raffel2020exploring}, the MMW Story dataset~\citep{piet2023mark}, and the Longform QA dataset from WaterBench~\citep{tu2023waterbench}. For text generation, we employ language models such as Llama-3.2-3B-Instruct~\citep{dubey2024llama} and Mistral-7B-Instruct-v0.3~\citep{jiang2023mistral}.

We run our experiments on 4 NVIDIA RTX 6000 Ada GPUs.

\section{Additional Experiments}
\subsection{Watermark Detectability under Different Length}
\label{appendix:res_vs_len}

Fig.~\ref{fig:res_vs_len} and \ref{fig:res_vs_len_full} illustrates how watermark detectability evolves as a function of text length across different datasets. We observe that weaker watermark configurations consistently dominate their stronger counterparts, highlighting the benefit of preserving token entropy. These trends are consistent across models, datasets and watermark schemes, demonstrating the robustness and generality of the proposed weaker ensemble watermarking framework.

\subsection{Ablation on the Number of Ensemble Layers}
\label{appendix:number_of_layers}
To mitigate the effect of the number of layers, we evaluate the True Positive Rate (TPR) across a varying number of ensemble layers for ENS-DiPmark, SynthID, and ENS-MCMark on both the Llama3-3B and Mistral-7B models.

As shown in Table~\ref{tab:layers_dipmark},\ref{tab:layers_synthid} and \ref{tab:layers_mcmark}, stronger watermarks consistently exhibit superior detectability in the initial layers. However, the performance of these stronger watermarks quickly plateaus or even degrades as the number of layers increases due to rapid entropy collapse.  In contrast, weaker watermarks preserve more entropy, allowing the ensemble to accumulate reliable statistical signals over a larger number of layers. Across all three methods, the weaker configurations ultimately overtake their stronger counterparts in later layers, consistently achieving a higher maximum detectability when an adequate number of ensemble layers is employed.

\subsection{Results on Larger Models}

To verify that the benefit of weaker ensemble watermarking is not limited to the 3B--7B backbones used in our main experiments, we further evaluate all three watermarking methods on \texttt{meta-llama/Meta-Llama-3-70B-Instruct}. Due to the inference cost of the 70B model, we use 96 samples from the MMW Story dataset and consider generation lengths of 150 and 250 tokens. We otherwise follow the same detection protocol and watermark-strength settings as in the main experiments. As shown in Table~\ref{tab:mmw_story_llama70b}, weakening the watermark consistently improves or matches detection performance on the 70B backbone.

\begin{table*}[t]
\centering
\caption{
Watermark detection performance with \texttt{meta-llama/Meta-Llama-3-70B-Instruct}
on the MMW Story dataset (96 samples) under different token budgets.
We report TPR at fixed FPR thresholds (0.1\%, 0.01\%, 0.001\%) and the
median detection p-value across samples. Lower p-values and higher TPR
indicate stronger detectability.
}
\label{tab:mmw_story_llama70b}
\setlength{\tabcolsep}{2pt}
\renewcommand{\arraystretch}{1.15}
\begin{tabular}{@{}l|cccc|cccc@{}}
\toprule
\multirow{3}{*}{Method} & \multicolumn{4}{c|}{150 tokens} & \multicolumn{4}{c}{250 tokens} \\ \cmidrule(l){2-9}
 & \multicolumn{3}{c|}{TPR@FPR} & \multirow{2}{*}{\begin{tabular}[c]{@{}c@{}}Median\\ p-value$\downarrow$\end{tabular}} & \multicolumn{3}{c|}{TPR@FPR} & \multirow{2}{*}{\begin{tabular}[c]{@{}c@{}}Median\\ p-value$\downarrow$\end{tabular}} \\
 & 0.1\%$\uparrow$ & 0.01\%$\uparrow$ & \multicolumn{1}{c|}{0.001\%$\uparrow$} &  & 0.1\%$\uparrow$ & 0.01\%$\uparrow$ & \multicolumn{1}{c|}{0.001\%$\uparrow$} &  \\ \midrule
ENS-DiPmark ($\alpha$=0.5) & 13\% & 4\% & \multicolumn{1}{c|}{1\%} & 7.71e-2 & 28\% & \textbf{17\%} & \multicolumn{1}{c|}{6\%} & 1.22e-2 \\
ENS-DiPmark ($\alpha$=0.4) & \textbf{17\%} & \textbf{8\%} & \multicolumn{1}{c|}{\textbf{2\%}} & \textbf{3.81e-2} & \textbf{34\%} & \textbf{17\%} & \multicolumn{1}{c|}{\textbf{10\%}} & \textbf{6.62e-3} \\ \midrule
SynthID ($\lambda$=1) & 51\% & 37\% & \multicolumn{1}{c|}{20\%} & 7.32e-4 & 76\% & 58\% & \multicolumn{1}{c|}{44\%} & 3.08e-5 \\
SynthID ($\lambda$=0.8) & \textbf{62\%} & \textbf{41\%} & \multicolumn{1}{c|}{\textbf{21\%}} & \textbf{2.29e-4} & \textbf{81\%} & \textbf{65\%} & \multicolumn{1}{c|}{\textbf{50\%}} & \textbf{9.33e-6} \\ \midrule
ENS-MCMark ($\lambda$=1) & 54\% & 36\% & \multicolumn{1}{c|}{16\%} & 6.74e-4 & 67\% & 58\% & \multicolumn{1}{c|}{45\%} & \textbf{1.90e-5} \\
ENS-MCMark ($\lambda$=0.8) & \textbf{56\%} & \textbf{38\%} & \multicolumn{1}{c|}{\textbf{25\%}} & \textbf{5.06e-4} & \textbf{72\%} & \textbf{60\%} & \multicolumn{1}{c|}{\textbf{46\%}} & \textbf{1.90e-5} \\ \bottomrule
\end{tabular}
\end{table*}

\begin{table*}[h]
\centering
\caption{
Watermark detection performance across five random seeds (42--46) using
Llama-3.2-3B-Instruct with a generation length of 150 tokens.
We report the mean and standard deviation of TPR at fixed FPR thresholds
(0.1\%, 0.01\%, and 0.001\%).
An asterisk ($^{*}$) denotes statistical significance.
}
\label{tab:seed_results_llama3b}
\setlength{\tabcolsep}{4pt}
\begin{tabular}{lccc}
\toprule
\textbf{Watermark}
& \textbf{TPR@FPR 0.1\%}
& \textbf{TPR@FPR 0.01\%}
& \textbf{TPR@FPR 0.001\%} \\
\midrule

ENS-DiPmark ($\alpha=0.5$)
& $51.7 \pm 6.5\%$
& $42.0 \pm 6.0\%$
& $34.5 \pm 5.6\%$ \\

ENS-DiPmark ($\alpha=0.4$)
& $\mathbf{53.7 \pm 7.4\%}$
& $\mathbf{44.5 \pm 6.7\%}^{*}$
& $\mathbf{36.5 \pm 5.9\%}^{*}$ \\

\midrule

SynthID ($\lambda=1$)
& $75.7 \pm 4.1\%$
& $68.0 \pm 4.4\%$
& $60.9 \pm 5.6\%$ \\

SynthID ($\lambda=0.8$)
& $\mathbf{77.7 \pm 4.4\%}^{*}$
& $\mathbf{71.0 \pm 6.0\%}^{*}$
& $\mathbf{64.3 \pm 6.3\%}^{*}$ \\

\midrule

ENS-MCMark ($\lambda=1$)
& $77.2 \pm 0.8\%$
& $70.6 \pm 1.5\%$
& $63.5 \pm 2.1\%$ \\

ENS-MCMark ($\lambda=0.8$)
& $\mathbf{79.0 \pm 1.1\%}^{*}$
& $\mathbf{72.4 \pm 1.1\%}^{*}$
& $\mathbf{66.3 \pm 1.9\%}^{*}$ \\

\bottomrule
\end{tabular}
\end{table*}

\subsection{Multi-seed Results}
To assess whether the gains from weaker watermarking are robust to randomness in generation and watermarking, we repeat the 150-token Llama-3.2-3B-Instruct experiment with five random seeds (42--46). Table~\ref{tab:seed_results_llama3b} reports the mean and standard deviation of TPR across seeds at the three FPR thresholds used throughout the paper.

\section{Comparison with SynthID}
\label{appendix:compare_synthid}

We clarify our contributions in the context of SynthID~\citep{dathathri2024scalable}. While they provide insights into a specific case, our work generalizes these phenomena into a universal theory applicable to all distortion-free ensemble watermarks.

Specifically, the scope of ~\citet{dathathri2024scalable} is limited to their proposed Tournament Sampling method (SynthID-Text). Their analytical findings rely on complex, method-specific derivations that are focused exclusively on collision entropy. In contrast, our framework applies broadly to all distortion-free ensemble watermarks. We provide a generalized proof that accommodates all concave entropies, rather than just collision entropy.

Furthermore, to analyze detectability, we introduce the concept of the expected green ratio. Using this concept, we theoretically demonstrate that the expected per-layer detectability will inevitably drop for all existing distortion-free ensemble watermarks because their expected green ratios are concave.

Finally, our work goes beyond theoretical analysis by offering a practical output. We propose a universal, applicable method to construct weaker distortion-free ensemble watermarks (as defined in Eq.~\ref{eq:weaker-watermark}). We empirically show that this approach yields significant improvements in both long-term detectability and robustness, a practical application not addressed in ~\citet{dathathri2024scalable}.

\end{document}